\documentclass[12pt]{article}

\usepackage{amssymb}
\usepackage{amsmath}
\usepackage{graphicx}%
\usepackage{multirow}%
\usepackage{amsfonts}%
\usepackage{amsthm}%
\usepackage{mathrsfs}%
\usepackage{xcolor}%
\usepackage{enumitem}

 \newtheorem{theorem}{Theorem}[section]
 \newtheorem{proposition}{Proposition}[section]
 \newtheorem{lemma}{Lemma}[section]
 \newtheorem{corollary}{Corollary}[section]
 \newtheorem{definition}{Definition}[section]
 \newtheorem{remark}{Remark}[section]
 \newtheorem{example}{Example}[section]

 \usepackage{graphicx}
 \usepackage{amsmath}\numberwithin{equation}{section}
 \usepackage{amssymb}
 \usepackage{epsfig}
\usepackage{amsfonts}
\usepackage{natbib}
\setcitestyle{authoryear,round}

\usepackage{mathrsfs}
\usepackage{amssymb,color}
\usepackage{caption2}
\usepackage{inputenc}

\begin{document}




%
\begin{center}
{\Large Non-emptiness of the Fuzzy Core for Exchange Economies with Discontinuous Utility}
\end{center}

\begin{center}
{Qi-Qing Song\footnote{Corresponding e-mail: songqiqing@126.com},\,Jia-Lin Liu,\,Jia-Yu Cui}\\
\end{center}

\begin{center}
{School of Mathematical Sciences, Shanxi Normal University, Taiyuan, 030031, China}\\
\end{center}


{\bf Abstract}
This paper investigates the non-emptiness of the fuzzy core and introduces the notion of the TU-fuzzy core for pure exchange economies with discontinuous utility functions. We propose two stability conditions: local fuzzy-coalition stability and local TU-fuzzy-coalition stability. We also introduce an independent-$\pi$-balanced family of subsets of the agent set to construct two payoff-dependent NTU games. These constructions serve as bridges linking the fuzzy core, the TU-fuzzy core, and NTU games. Under the proposed stability conditions, we establish the non-emptiness of the fuzzy core for economies with quasi-concave utility functions and the non-emptiness of the TU-fuzzy core for economies with concave utility functions.

{\it Keyword:}
Fuzzy core, exchange economy, discontinuous, payoff-dependent, independent-$\pi$-balanced, existence


\section{Introduction}

The notion of the fuzzy core originates with \citet{aubin1979}, who introduced it as a cooperative solution concept for pure exchange economies that generalizes the classical core. In the standard formulation of \citet{edgeworth1881} and \citet{debreu1963}, coalitional blocking requires the unanimous and full participation of every member. Aubin relaxed this all-or-nothing requirement by permitting agents to join a blocking coalition with an arbitrary fraction of their resources, indexed by a coefficient $\alpha_i \in [0,1]$. A fuzzy coalition is therefore represented by a vector $\alpha = (\alpha_1, \ldots, \alpha_n) \in [0,1]^n \setminus \{\mathbf{0}\}$, where $\alpha_i$ captures the intensity of agent $i$'s participation. The fuzzy core then consists of all feasible allocations that are robust against objections raised by any such fractional coalition.

The conceptual significance of the fuzzy core lies in its capacity to deliver a cooperative characterization of competitive equilibrium without recourse to the asymptotic reasoning central to the classical Debreu-Scarf approach. Aubin demonstrated that, under standard convexity assumptions, the fuzzy core of a finite economy coincides exactly with the set of competitive equilibria. As \citet{herves2026} have recently emphasized, this equivalence can be reinterpreted through the lens of the classical Debreu-Scarf and Aumann theorems: embedding an $n$-agent economy into a continuum economy composed of $n$ distinct agent types, the fuzzy core maps precisely onto the set of Edgeworth allocations of the continuum economy, which in turn coincides with the Walrasian allocations.

The existence of a nonempty fuzzy core has been established under increasingly general conditions. Aubin's original theorem required convexity of preferences and finite-dimensional commodity spaces. \citet{florenzano2003} significantly extended this result to infinite-dimensional commodity spaces and to economies with a non-trivial production sector, while relaxing the assumptions on consumer preferences by allowing them to be incomplete or intransitive. \citet{predtetchinski2005fuzzy} offered an alternative proof for economies with continuous and quasi-concave utility functions, based on the $(\Pi,\beta)$-balanced core for non-transferable utility (NTU) games (closely related to the results in \citet{Hering2004necessary}), an important case of the game introduced by \citet{bonnisseau2003payoff,bonnisseau2007payoff}. The proof proceeds by associating with the exchange economy an NTU game whose $(\Pi,\beta)$-balanced core is nonempty; this balanced core is then shown to correspond to the $\epsilon$-fuzzy core of the economy, and the non-emptiness of the fuzzy core itself follows by passing to the limit as $\epsilon \to 0$. \citet{allouch2008} later provided an elementary fixed-point proof that avoids any reliance on asymptotic intersections or limits of replica-economy cores, for economies with continuous and quasi-concave utility functions. By means of a necessary and sufficient condition for fuzzy NTU games, \citet{liu2013} proved the non-emptiness of the fuzzy core for an economy with continuous and concave utility functions. More recently,  \citet{marakulin2022} contributed yet another proof based on Michael's continuous selection theorem and Brouwer's fixed point theorem, requiring only minimal assumptions that accommodate incomplete, non-transitive, or satiated preferences.

For economies with upper semicontinuous utility functions, \citet{liu2022,liu2023} proved the non-emptiness of the fuzzy core in finite coalition production economies with infinite-dimensional commodity spaces, thus generalizing earlier work by \citet{predtetchinski2005fuzzy}. To further study the non-emptiness of the fuzzy core of an exchange economy with discontinuous utility functions, we intend to recall some closely related results on the existence of equilibrium points for games with discontinuous payoffs, such as $\alpha$-cores and Nash equilibrium points.

The existence of the $\alpha$-core (proposed by \citet{aumann1961}) for discontinuous normal-form games was first established in \citet{uyanik2015}, where the notion of coalitional $C$-security plays a pivotal role. In the same paper, Uyan{\i}k (2015) also obtained the existence of the $\alpha^T$-core under coalitional $C^T$-security together with a weak separability condition, which is closely related to the weak separability introduced in \citet{zhao1999b}. The coalitional $C$-security proposed by \citet{uyanik2015} draws on various security concepts from discontinuous noncooperative games, including better-reply security (\citet{reny1999}), $C$-security (\citet{mclennan2011}), and continuous security (\citet{barelli2013}). More recently, \citet{scalzo2020ijgt} introduced the coalitional deviation property to establish a necessary and sufficient condition for the non-emptiness of the $\alpha$-core, while \citet{scalzo2020existence} proposed the slight single deviation property to prove the existence of strong equilibria. For further results on the existence of Nash equilibria with discontinuous payoffs, we refer to \citet{Prokopovych2011,NessahTian2016} and the references therein. The non-emptiness of the $\alpha$-core in games with discontinuous payoffs can be found in \citet{yang2018fuzzygames,jia2023fuzzy,song2024structure,song2024alpha,chen2024existence,Chi2026}, among others.

The stability conditions for deriving the above equilibrium results under discontinuous payoffs are instructive. These stability conditions, especially the coalitional $C$-security and the coalitional deviation property, motivate us to propose new conditions for the fuzzy core.

In this paper, we further study the non-emptiness of the fuzzy core of an exchange economy. The concept of the TU-fuzzy core of an exchange economy is introduced. A locally fuzzy-coalition stable condition and a locally TU-fuzzy-coalition stable condition for economies are proposed. By introducing an independent-$\pi$-balanced family of subsets of the set of agents, two payoff-dependent cores of NTU games are proved, which bridge the fuzzy core and the TU-fuzzy core of exchange economies.

The rest of the paper is organized as follows. Section \ref{sec-2} recalls the concepts of the fuzzy core of an exchange economy and the payoff-dependent core of NTU games. Furthermore, an example shows that the fuzzy core of an exchange economy can be nonempty even when the utility functions are neither upper semi-continuous nor lower semi-continuous. Section \ref{sec-3} introduces the concepts of an independent-$\pi$-balanced family and TU-fuzzy cores, as well as locally (TU-)fuzzy-coalition stable conditions. The non-emptiness of the fuzzy core and the TU-fuzzy core is established for an exchange economy.

\section{Fuzzy core and payoff-dependent balancedness}\label{sec-2}

\subsection{Fuzzy core of exchange economies}

In a market, let $N=\{1,2,\cdots,n\}$ be the set of agents, and let $\mathcal{N}$ be the set of nonempty subsets of $N$. The commodity space is denoted by $\mathrm{L}$, which is a real linear Hausdorff topological space. The triple $\mathcal{E}=\{X_{i},w_i,u_i\}_{i\in N}$ defines a pure exchange economy, in which, for each agent $i\in N$, the consumption set of the agent is $X_{i}$, a subset of $\mathrm{L}$; the initial endowment of commodities is denoted by $w_{i}\in \mathrm{L}$; and the utility function of the agent is given by $u_{i}: X_{i}\to \mathbb{R}$.

For each $S\in \mathcal{N}$, let $X_{S}=\times_{i\in S}X_{i}$, and let $x_S$ denote a joint status in $X_S$ with $x_S=(x_{S,i})_{i\in S}$ or $x_S=(x_{i})_{i\in S}$. For short, denote $X_N$ by $X$. In order to define fuzzy cores and some balancedness notions in an economy, we need some notation. For $S\in \mathcal{N}$, let $\Delta_{S} = \{\pi\in \mathbb{R}^{n} :  \pi_{i}\geq 0 \text{ for all } i\in S \text{ and } \sum_{i\in S}\pi_{i} = 1\}$. Clearly, $\Delta_{N}$ is the unit simplex in $\mathbb{R}^{n}$, and $m_N$ denotes the barycenter of $\Delta_{N}$, that is, $m_N=(\frac{1}{n},\frac{1}{n},\cdots,\frac{1}{n})\in \Delta_{N}$. Here, for each $\pi\in \Delta_{N}$, $\pi$ is a fuzzy coalition, and each $\pi_i$ denotes the participation rate of agent $i\in N$. $\Delta_{N}$ denotes all possible fuzzy coalitions. For each $\pi\in \Delta_N$, the symbol $\operatorname{carr}(\pi)$ denotes the carrier of $\pi$, that is, $\operatorname{carr}(\pi)=\{i\in N: \pi_i>0\}$.

We consider the following two assumptions:
\begin{enumerate}[label=(A\arabic*), ref=(A\arabic*)]
\item\label{items:A1} For each $i\in N$, the set $X_{i}$ is a closed convex set with $w_{i}\in X_i$.
\item\label{items:A2} The set $F(\mathcal{E}) = \{x\in X :  \sum_{i\in N}x_{i} = \sum_{i\in N}w_{i}\}$ of feasible commodity allocations is compact.
\end{enumerate}

Note that under Assumptions \ref{items:A1} and \ref{items:A2}, the set of feasible commodity allocations, $F(\mathcal{E})$, is compact and convex.

\begin{definition}\label{def0}
Given an economy $\mathcal{E}=\{X_{i},w_i,u_i\}_{i\in N}$, an allocation $x \in X$ is blocked by $S \subseteq N$ if there exist $\pi \in \Delta_S$ with $\pi_{i}>0$ for all $i\in S$, and $y \in X_S$ with $\sum_{i \in S} \pi_{i} y_{i} = \sum_{i \in S} \pi_i w_i$, such that $u_{i} (y_{i}) > u_{i} (x_i)$ for all $i \in S$.
\end{definition}

\begin{remark}
In Definition \ref{def0}, the allocation $x$ is also said to be blocked by a fuzzy coalition $\pi \in \Delta_S$ with carrier $S$.
\end{remark}

\begin{definition}\label{def1}
Given an economy $\mathcal{E}=\{X_{i},w_i,u_i\}_{i\in N}$, an allocation $x \in F(\mathcal{E})$ is an element of the fuzzy core $C_f(\mathcal{E})$ of the economy $\mathcal{E}$ if there exists no $S \in \mathcal{N}$ that blocks the allocation $x$.
\end{definition}

From Definition \ref{def1}, no fuzzy coalition $\pi$ with any $S\in \mathcal{N}$ as its carrier can block the allocation $x \in C_f(\mathcal{E})$. That is, there is no fuzzy coalition $\pi\in \Delta_N$ that can block the allocation. Thus, we find that Definition \ref{def1} is equivalent to the following Definition \ref{def2} in \citet{predtetchinski2005fuzzy} or \citet{aubin1979}.

\begin{definition}\label{def2}
An allocation $x \in F(\mathcal{E})$ is an element of the fuzzy core $C_f(\mathcal{E})$ of the economy $\mathcal{E}$ if there exist no $\pi \in \Delta_N$ and no $y \in X_N$ such that $\sum_{i \in N} \pi_{i} y_{i} = \sum_{i \in N} \pi_i w_i$ and $u_{i} (y_{i}) > u_{i} (x_i)$ for all $i \in \operatorname{carr}(\pi)$.
\end{definition}

The following simple Example \ref{exam-1} shows a specific economy with discontinuous utility functions (neither upper semi-continuous nor lower semi-continuous). However, it has a nonempty fuzzy core.

\begin{example}\label{exam-1}
Let \( L = \mathbb{R}^2 \), \( N = \{1,2\} \), $X_1 = X_2 = [0,2] \times [0,2] \subset \mathbb{R}^2$, $w_1 = (1,0), \,w_2 = (0,1)$, and consider the pure exchange economy
$
\mathcal{E} = \{(X_i, w_i, u_i)\}_{i \in N}
$
defined with utility functions as follows:
\begin{align*}
u_1(x_1) &=
\begin{cases}
1, & x_1 = (1,0), \\
0, & \text{otherwise},
\end{cases} \\
u_2(x_2) &=
\begin{cases}
1, & x_2 = (0,1), \\
-1, & x_2 = (0,0), \\
0, & \text{otherwise}.
\end{cases}
\end{align*}

Clearly, each \( X_i \) is closed and convex, and \( w_i \in X_i \). The feasible set
\[
F(\mathcal E) = \left\{ (x_1,x_2) \in X_1 \times X_2 : x_1 + x_2 = (1,1) \right\}
\]
is compact. Note that \( u_2 \) is not continuous. In fact, it is neither upper semi-continuous nor lower semi-continuous on $X_2$. Note also that both $u_1$ and $u_2$ are quasi-concave.

Let us compute the fuzzy core using Definition \ref{def2}.
Let
\[
\bar{x} = ((1,0), (0,1)) = (w_1, w_2) \in F(\mathcal E).
\]
At the allocation \( \bar{x} \), we have
\(
u_1(w_1) = 1, \qquad u_2(w_2) = 1.
\)

For any other feasible allocation \( x \in F(\mathcal E) \setminus \{\bar{x}\} \), we have \( x_1 \neq w_1 \) or \( x_2 \neq w_2 \). Note that \( x \in F(\mathcal E)\). Hence
\[
u_1(x_1) \le 0 \quad \text{and} \quad u_2(x_2) \le 0.
\]

Now we verify that \( \bar{x} \) is in the fuzzy core. Suppose, to the contrary, that there exist \( \pi \in \Delta_N \) and \( y \in X \) such that
\[
\sum_{i=1}^2 \pi_i y_i = \sum_{i=1}^2 \pi_i w_i
\]
and
\[
u_i(y_i) > u_i(\bar{x}_i) = 1 \quad \text{for all } i \in \operatorname{carr}(\pi).
\]
But \( u_i \) attains its maximum value \( 1 \) only at \( w_i \) for each \( i \). Thus the inequalities \( u_i(y_i) > 1 \) are impossible. Hence \( \bar{x} \) belongs to the fuzzy core.

Conversely, take any \( x \in F(\mathcal E) \setminus \{\bar{x}\} \). Since \( u_1(x_1) \le 0 \) and \( u_2(x_2) \le 0 \), we can choose \( \pi = (1/2, 1/2) \) and \( y = \bar{x} \). Then
\[
\sum_{i=1}^2 \pi_i y_i = \frac12 w_1 + \frac12 w_2 = \sum_{i=1}^2 \pi_i w_i,
\]
and
\[
u_1(y_1) = 1 > u_1(x_1), \qquad u_2(y_2) = 1 > u_2(x_2).
\]
Thus, \( x \) is not in the fuzzy core. Therefore,
\(
C_f(\mathcal E) = \{\bar{x}\}.
\)
Hence, the economy with these discontinuous utility functions has a nonempty fuzzy core.
\end{example}

\subsection{The core related to payoff-dependent balancedness}

A characteristic function game $V = ( V(S) )_{S \in \mathcal{N}}$ with non-transferable utility is a collection of subsets of $\mathbb{R}^n$, where $V(S) \subset \mathbb{R}^n$ is the (feasible) payoff set of crisp coalition $S$. Let $W = \bigcup_{S \in \mathcal{N}} V(S)$. For a subset $A \subset \mathbb{R}^n$, the symbols $\operatorname{int} A$, $\partial A$, and $\operatorname{co} A$ denote, respectively, the interior, the boundary, and the convex hull of $A$. The following description of the core of a characteristic function game is from \citet{bonnisseau2007payoff}.

\begin{definition}\label{def5}
Let $(V(S))_{S \in \mathcal{N}}$ be a NTU game. The core, $C(V) \subset \mathbb{R}^n$, of the game is given by
\[
C(V) = V(N) \setminus \operatorname{int} W.
\]
\end{definition}

In the current paper, the following assumptions are also considered. For each $S \in \mathcal{N}$,
\begin{enumerate}[label=(B\arabic*), ref=(B\arabic*)]
\item \label{item:B1} $V(S)$ is a nonempty closed subset of $\mathbb{R}^n$.
\item \label{item:B2} $V(S)$ is comprehensive, that is, if $v \in V(S)$, $v' \in \mathbb{R}^n$, and $v'_i \leq v_i$ for all $i \in S$, then $v' \in V(S)$.
\item \label{item:B3} There exists a number $M$ such that, for each $v \in V(S)$, if $v \notin \operatorname{int} V(\{i\})$ for all $i \in S$, then $v_j \leq M$ for all $j \in S$.
\end{enumerate}

\begin{proposition}\label{pro_1}
Let $(V(S))_{S \in \mathcal{N}}$ be a game satisfying Assumptions \ref{item:B1} and \ref{item:B2}. Then $\operatorname{int} W = \bigcup_{S \in \mathcal{N}} \operatorname{int} V(S)$.
\end{proposition}
\begin{proof}
Since $V(S) \subseteq W$, it is clear that $\bigcup_{S \in \mathcal{N}} \operatorname{int} V(S) \subseteq \operatorname{int} W$. Take $v \in \operatorname{int} W \subset W$. Then, there exists $\varepsilon > 0$ such that $B(v,\varepsilon) \subset \operatorname{int} W \subset W$, where $B(v,\varepsilon)$ is the ball centered at $v$ with radius $\varepsilon$. Then, there exists $\delta > 0$ such that $v + \delta \mathbf{1} \in B(v,\varepsilon) \subset W = \bigcup_{S \in \mathcal{N}} V(S)$, where $\mathbf{1}$ is the vector in $\mathbb{R}^n$ whose coordinates are all equal to $1$. There is some $S' \in \mathcal{N}$ such that $v + \delta \mathbf{1} \in V(S')$. Since $V(S')$ is comprehensive, we have $v + \delta \mathbf{1} - \mathbb{R}^n_+ \subset V(S')$. Then, there exists $\delta' > 0$ with $\delta' < \delta$ such that $v \in B(v,\delta') \subset \operatorname{int} V(S') \subset \bigcup_{S \in \mathcal{N}} \operatorname{int} V(S)$. The result follows.
\end{proof}

By Proposition \ref{pro_1}, under Assumptions \ref{item:B1} and \ref{item:B2}, the core in Definition \ref{def5} is equivalent to the core in Definition \ref{def6}.

\begin{definition}\label{def6}
Let $(V(S))_{S \in \mathcal{N}}$ be a game. The core, $C(V) \subset \mathbb{R}^n$, of the game is given by
\[
C(V) = V(N) \setminus \bigcup_{S \in \mathcal{N}} \operatorname{int} V(S).
\]
\end{definition}

For any $v \in \mathbb{R}^n$, let $\mathcal{S}(v) = \{S \in \mathcal{N}  :  v \in \partial V(S)\}$. Payoff-dependent balanced games were introduced in \citet{bonnisseau2007payoff}. They were defined by means of the notion of transfer rate rules. Under Assumptions \ref{item:B1} and \ref{item:B2}, we present them as Definitions \ref{def7} and \ref{def8} below.

\begin{definition}\label{def7}
Let $(V(S), S \in \mathcal{N})$ be a game. A transfer rate rule is a collection of set-valued mappings $((\phi_S)_{S \in \mathcal{N}}, \psi)$ such that: for each $S \in \mathcal{N}$, $\phi_S: \partial V(S) \rightrightarrows \Delta_S$ is upper semi-continuous with nonempty compact and convex values; and $\psi: \partial V(N) \rightrightarrows \Delta_N$ is upper semi-continuous with nonempty compact and convex values.
\end{definition}

As stated in \citet{bonnisseau2007payoff}, a game is payoff-dependent balanced if it satisfies a condition with respect to a transfer rate rule $((\phi_S)_{S \in \mathcal{N}}, \psi)$. For convenience, we consider a special case (Definition \ref{def8} below) where $\psi(p_N(v)) \equiv m_N$ for all $v \in \mathbb{R}^n$, with $p_N$ being a continuous mapping from $\mathbb{R}^n$ to $\partial V(N)$. In the remainder of the paper, $\psi$ is taken to be $m_N$.

\begin{definition}\label{def8}
The game $(V(S), S \in \mathcal{N})$ is payoff-dependent balanced with respect to the transfer rate rule $((\phi_S)_{S \in \mathcal{N}}, m_N)$ if, for each $v \in \partial W$,
\[
m_N \in \operatorname{co}\{\phi_S(v)  :  S \in \mathcal{S}(v)\} \Rightarrow v \in V(N).
\]
\end{definition}

\begin{remark}
In \citet{predtetchinski2005fuzzy}, the payoff-dependent balancedness in Definition \ref{def8} is called $(\phi, m_N)$-balancedness, following the original reference \citet{bonnisseau2003payoff}.
\end{remark}

A result on the non-emptiness of the core of a payoff-dependent balanced game was given in \citet{bonnisseau2007payoff} as Lemma \ref{lem_1} below.

\begin{lemma}\label{lem_1}
Let $(V(S), S \in \mathcal{N})$ be a game satisfying Assumptions \ref{item:B1}--\ref{item:B3}. If the game is payoff-dependent balanced with respect to the transfer rate rule $((\phi_S)_{S \in \mathcal{N}}, m_N)$, then there exists a payoff $v$ in the core such that
\[
m_N \in \operatorname{co}\{\phi_S(v)  :  S \in \mathcal{S}(v)\}.
\]
\end{lemma}

From Lemma \ref{lem_1}, we see that a payoff-dependent balanced game has at least one core payoff vector satisfying an additional requirement, which provides a method of core selection.

\section{Non-emptiness of (TU) fuzzy cores for economies with discontinuous utility functions}\label{sec-3}
\subsection{Non-emptiness of the fuzzy core}

We introduce a new balanced family of nonempty subsets of $N$ as in Definition \ref{def4} below.

\begin{definition}\label{def4}
Let $K$ be a finite index set. For each $S \in \mathcal{N}$ and each $k \in K$, let $\pi_{S}^k \in \Delta_S$. A family $\beta\subseteq\mathcal{N}$ is called \emph{independent-$\pi$-balanced} if there exist nonnegative numbers $\lambda_S^k \ge 0$ for $S\in\beta$ and $k\in K$ such that
\[
\sum_{S \in \beta,\, k \in K} \lambda_S^k \pi_S^k = m_N,
\qquad
\sum_{S \in \beta,\, k \in K} \lambda_S^k = 1,
\]
and the vectors
\[
\{\pi_S^k \in \mathbb{R}^n : S\in\beta,\ k\in K,\ \lambda_S^k > 0\}
\]
are affinely independent.
\end{definition}
\begin{remark}
$(i)$ The number $\lambda_S^k$ is called the balanced weight corresponding to coalition $S$ and index $k$ in Definition \ref{def4}.

$(ii)$ For a fixed $\beta\subseteq\mathcal N$, the number of weight vectors $\lambda=(\lambda_S^k)_{S\in\beta,k\in K}$ satisfying the independent-$\pi$-balanced condition is finite.

$(iii)$ Given only one fixed $\pi_S \in \Delta_S$ for each $S \in \beta$, we obtain a vector $\pi = (\pi_S)_{S \in \beta}$. For this fixed $\pi$, if the family $\beta$ is independent-$\pi$-balanced according to Definition \ref{def4}, then it is $\pi$-balanced in the sense of \citet{shapley1973balanced}.
\end{remark}

To establish the non-emptiness of fuzzy core, we introduce local fuzzy-coalition stability for a pure exchange economy as follows.

\begin{definition}\label{def9}
A pure exchange economy $\mathcal{E} = \{X_{i}, w_i, u_i\}_{i \in N}$ is locally fuzzy-coalition stable if, for each $x \in F(\mathcal{E})$ that is not in the fuzzy core of $\mathcal{E}$, there exist $\pi_S(x) \in \Delta_S$ with $\pi_{S,i}(x) > 0$ for all $i \in S$, $z_S(x) \in X_S$ with $\sum_{i \in S} \pi_{S,i}(x) z_{S,i}(x) = \sum_{i \in S} \pi_{S,i}(x) w_i$ for each $S \in \mathcal{N}$, and an open neighborhood $O(x)$ of $x$ in $F(\mathcal{E})$, such that for each $y \in O(x) \setminus C_f(\mathcal{E})$, there exists $S \in \mathcal{N}$ for which $u_i(z_{S,i}(x)) > u_i(y_i)$ for all $i \in S$.
\end{definition}

Example \ref{exam-1} shows that there exists an economy with discontinuous utility functions (neither upper semi-continuous nor lower semi-continuous) that nevertheless has a nonempty fuzzy core. We now explain why the fuzzy core of the economy in Example \ref{exam-1} is nonempty by showing in Example \ref{exam-2} below that the economy is locally fuzzy-coalition stable.

\begin{example}\label{exam-2}
Let $\mathcal{E}$ be the economy in Example \ref{exam-1}. Take any $x \in F(\mathcal E) \setminus C_f(\mathcal E)$. Since $C_f(\mathcal E) = \{\bar{x}\}$ is a singleton, there exists an open neighborhood $O(x)$ of $x$ in $F(\mathcal E)$ such that $\bar{x} \notin O(x)$.

Fix the crisp grand coalition $N = \{1,2\}$, choose
\[
\pi_N(x) = \left( \frac12, \frac12 \right) \in \Delta_N,
\]
and choose
\[
z_N(x) = \bar{x} = ((1,0), (0,1)) \in X.
\]
Note that the constraint requirement holds because
\[
\sum_{i \in N} \pi_{N,i}(x) z_{N,i}(x)
= \frac12 (1,0) + \frac12 (0,1)
= \left( \frac12, \frac12 \right)
= \frac12 w_1 + \frac12 w_2
= \sum_{i \in N} \pi_{N,i}(x) w_i.
\]

Now, for any $y \in O(x)$, since $y \neq \bar{x}$, we have $y_1 \neq (1,0)$ or $y_2 \neq (0,1)$. Therefore,
\[
u_1(y_1) \le 0 \quad \text{and} \quad u_2(y_2) \le 0.
\]
But
\[
u_1(z_{1}) = u_1(1,0) = 1, \qquad u_2(z_{2}) = u_2(0,1) = 1.
\]
Hence
\[
u_1(z_{1}) = 1 > u_1(y_1), \qquad u_2(z_{2}) = 1 > u_2(y_2).
\]
Thus, for every $y \in O(x)$, the coalition $N = \{1,2\}$ strictly improves the utilities of both agents.

Therefore, by Definition \ref{def9}, the economy $\mathcal{E}$ is locally fuzzy-coalition stable. This will explain why, despite the discontinuities of the utility functions, the economy has a nonempty fuzzy core, as shown in Theorem \ref{thm3-1} below.
\end{example}

Local fuzzy-coalition stability is guaranteed if all utility functions in the economy are upper semi-continuous. Therefore, the economy with locally fuzzy-coalition-stable utility functions includes a large class of economies.

\begin{proposition}\label{pro_2}
For an economy $\mathcal{E} = \{X_{i}, w_i, u_i\}_{i \in N}$, if the utility function $u_i$ is upper semi-continuous for each $i \in N$, then $\mathcal{E}$ is locally fuzzy-coalition stable.
\end{proposition}
\begin{proof}
If $x \in F(\mathcal{E})$ and $x$ is not in the fuzzy core of $\mathcal{E}$, then there exist $S \in \mathcal{N}$, $\pi \in \Delta_S$ with $\pi_i > 0$ for all $i \in S$, and $z \in X_S$ with $\sum_{i \in S} \pi_i z_i = \sum_{i \in S} \pi_i w_i$, such that $u_i(z_i) > u_i(x_i)$ for all $i \in S$. Since $u_i$ is upper semi-continuous for each $i \in S$, there exists an open neighborhood $U^i(x)$ in $F(\mathcal{E})$ such that $u_i(z_i) > u_i(y_i)$ for all $y \in U^i(x)$. Let $O(x) = \cap_{i \in S} U^i(x)$. Then, $u_i(z_i) > u_i(y_i)$ for all $y \in O(x)$.
\end{proof}

To establish the non-emptiness of the fuzzy core of an exchange economy, we need the following lemma.

\begin{lemma}\label{lem3-1}
Let $T$ be a finite index set. Given $S \in \mathcal{N}$, let $b^k \in \mathbb{R}^{|S|}$ and $\pi^k \in \Delta_S$ for each $k \in T$. Define
\[
A = \bigcup_{k \in T} \{ a \in \mathbb{R}^n : b_i^k \ge a_i,\ \forall i \in S \},
\]
and $\phi: \partial A \rightrightarrows \Delta_S$ by
\[
\phi(a) = \operatorname{co}\{ \pi^k \in \Delta_S  :  b_i^k \ge a_i,\ \forall i \in S \}.
\]
Then, the set-valued mapping $\phi$ is upper semi-continuous on $\partial A$.
\end{lemma}

\begin{proof}
Clearly, $\phi$ has nonempty values. Since $\Delta_S$ is compact, it suffices to prove that the graph of $\phi$ is closed.

Let $\{a^m\} \subset \partial A$ be a sequence such that $a^m \to a^0 \in \partial A$. Take any $y^m \in \phi(a^m)$ and suppose $y^m \to y^0 \in \mathbb{R}^n$. It needs to show $y^0 \in \phi(a^0)$.

For each $m$, by the definition of $\phi$, the point $y^m$ lies in the convex hull of the finite set $\{\pi^k : k \in T\}$. Hence, there exist coefficients $\lambda^m_{k} \ge 0$ for $k \in T$ such that
\[
\sum_{k \in T} \lambda^m_{k} = 1,
\quad
y^m = \sum_{k \in T} \lambda^m_{k} \pi^k,
\]
and the following support condition holds:
\begin{equation}\label{eq-16}
\lambda^m_{k} > 0 \Rightarrow b_i^k \ge a^m_i,\ \forall i \in S.
\end{equation}

Define the simplex
\(
\Delta_T = \left\{ \lambda \in \mathbb{R}^{|T|}_+  :  \sum_{k \in T} \lambda_k = 1 \right\}.
\)
Since $T$ is finite, $\Delta_T$ is compact in $\mathbb{R}^{|T|}$. Therefore the sequence $\{\lambda^m\} = \{(\lambda^m_{k})_{k \in T}\}$ admits a convergent subsequence. Relabeling this subsequence as $\{\lambda^m\}$, we may assume
\(
\lambda^m \to \lambda^0 = (\lambda^0_{k})_{k \in T} \in \Delta_T.
\)

Since $y^m \to y^0$, we obtain
\begin{equation}\label{eq-17}
y^0 = \lim_{m \to \infty} y^m
= \lim_{m \to \infty} \sum_{k \in T} \lambda^m_{k} \pi^k
= \sum_{k \in T} \lambda^0_{k} \pi^k.
\end{equation}

We now claim that for every $k \in T$ with $\lambda^0_{k} > 0$, it holds that
\begin{equation}\label{eq-18}
b_i^k \ge a^0_i,\quad \forall i \in S.
\end{equation}

Indeed, fix $k \in T$. Since $\lambda^m_{k} \to \lambda^0_{k} > 0$, there exists an integer $M$ such that $\lambda^m_{k} > 0$ for all $m \ge M$. By the support condition \eqref{eq-16}, for all $m \ge M$ and all $i \in S$, we have $b_i^k \ge a^m_i$. As $m \to \infty$, $a^m_i \to a^0_i$, hence $b_i^k \ge a^0_i$ for all $i \in S$. This proves \eqref{eq-18}.

Define the index set
\[
I(a^0) := \{ k \in T : b_i^k \ge a^0_i,\ \forall i \in S \}.
\]
By \eqref{eq-18}, every $k \in T$ with $\lambda^0_{k} > 0$ belongs to $I(a^0)$. Consequently, from \eqref{eq-17}, we have
\[
y^0 = \sum_{k \in T} \lambda^0_{k} \pi^k
= \sum_{k \in I(a^0)} \lambda^0_{k} \pi^k + \sum_{k \notin I(a^0)} \lambda^0_{k} \pi^k.
\]
The second term on the right-hand side vanishes since all its coefficients are zero. Thus
\(
y^0 = \sum_{k \in I(a^0)} \lambda^ 0_{k} \pi^k.
\)
Moreover, $\sum_{k \in I(a^0)} \lambda^0_{k} = 1$ because $\sum_{k \in T} \lambda^0_{k} = 1$. Therefore
\[
y^0 \in \operatorname{co}\{\pi^k : k \in I(a^0)\} = \phi(a^0).
\]
This completes the proof.
\end{proof}

By constructing a special transfer rate rule for an NTU game, the first main result (Theorem \ref{thm3-1}) proves the non-emptiness of the fuzzy core for a locally fuzzy-coalition-stable economy with quasi-concave utility functions.

\begin{theorem}\label{thm3-1}
Let $\mathcal{E} = \{X_{i}, w_i, u_i\}_{i \in N}$ be a pure exchange economy satisfying local fuzzy-coalition stability and Assumptions \ref{items:A1} and \ref{items:A2}. For each agent $i \in N$, assume that the utility function $u_i$ is quasi-concave. Then, the fuzzy core of the economy is nonempty.
\end{theorem}
\begin{proof}
Suppose $C_f(\mathcal{E}) = \emptyset$. Then, for each $x \in F(\mathcal{E})$ and each $S \in \mathcal{N}$, there exist $\pi_S(x) \in \Delta_S$ with $\pi_{S,i} > 0$ for all $i \in S$, $z_S(x) \in X_S$, and an open neighborhood $O(x)$ of $x$ in $F(\mathcal{E})$ satisfying Definition \ref{def9}. Since $F(\mathcal{E})$ is compact, there exists a finite set $\{x^k  :  k \in K\} \subset X$ such that the union $\bigcup_{k \in K} O(x^k)$ is a subcover of $\bigcup_{x \in F(\mathcal{E})} O(x)$. For convenience, we denote $z_S(x^k)$ by $z_S^k$ and $\pi_S(x^k)$ by $\pi_S^k$ for all $k \in K$ and $S \in \mathcal{N}$. For these $\pi_S^k$ for all $S \in \mathcal{N}$ and $k \in K$, let $\mathcal{B}$ be the collection of all independent-$\pi$-balanced families (in Definition \ref{def4}) of $\mathcal{N}$.

For any $S \in \mathcal{N}$ with $S \neq N$, we define
\begin{equation}\label{eq-1}
V(S) = \bigcup_{k \in K} \{ a \in \mathbb{R}^n : u_i(z_{S,i}^k) \ge a_i,\ \forall i \in S \},
\end{equation}
and
\begin{equation}\label{eq-2}
V(N) = \bigcup_{\beta \in \mathcal{B}} \{ a \in \mathbb{R}^n : u_i(x^\beta_i) \ge a_i,\ \forall i \in N \},
\end{equation}
where, for any given $\beta$, we require that
\begin{equation}\label{eq-3}
x_i^\beta = \sum_{S \in \beta(i), k \in K} n \lambda_S^k(\beta) \pi^{k}_{S,i} z_{S,i}^{k}, \quad \forall i \in N,
\end{equation}
and $\beta(i) = \{ S \in \beta  :  i \in S \}$.
Note that in Eq.~\eqref{eq-3}, $\lambda_S^k(\beta) \ge 0$ ranges over all finite balanced weights with respect to the independent-$\pi$-balanced family $\beta$. Since the family $\beta$ is independent-$\pi$-balanced, by Definition \ref{def4}, we have
\begin{equation}\label{eq-13}
\sum_{S \in \beta(i), k \in K} n \lambda_S^k(\beta) \pi^{k}_{S,i} = 1, \quad \forall i \in N.
\end{equation}
Note that $z_{S,i}^{k} \in X_i$ for all $i \in N$ and $k \in K$. Then, $x_i^\beta \in X_i$ for all $i \in N$.

Thus, from Eqs.~\eqref{eq-1}--\eqref{eq-3}, a characteristic function game $V$ is constructed.

For each $S \in \mathcal{N}$ with $S \neq N$, by its construction, the set $V(S)$ is nonempty, closed in $\mathbb{R}^n$, and comprehensive. Since $|\mathcal{B}|$ is finite and the set $\{u_i(x^\beta_i)  :  \beta \in \mathcal{B}\}$ is also finite by the finiteness of balanced weights for each $\beta\in \mathcal{B}$, $V(N)$ is nonempty, closed, and comprehensive. That is, the game $V$ satisfies Assumptions \ref{item:B1} and \ref{item:B2}. From the finiteness of elements in $K$ and $\mathcal{B}$, $V(S)$ is bounded above for each $S \in \mathcal{N}$; that is, there exists a number $M$ such that if $v \in V(S)$, then $v_i \leq M$ for all $i \in S$. Hence, the game $V$ satisfies Assumption \ref{item:B3}.

Next, we prove that the game $V$ satisfies payoff-dependent balancedness in Definition \ref{def8}. For each $S \in \mathcal{N}$, define $\phi_S : \partial V(S) \rightrightarrows \Delta_S$ by
\begin{equation}\label{eq-15}
\phi_S(a) = \operatorname{co}\{ \pi_S^{k} \in \Delta_S  :   k \in K \text{ such that } u_i(z_{S,i}^k) \ge a_i,\ \forall i \in S \}
\end{equation}
for each $a \in \partial V(S)$. By Lemma \ref{lem3-1}, $\phi_S$ is upper semi-continuous with nonempty compact and convex values. For each $a \in \partial W$, if $m_N \in \operatorname{co}\{\phi_S(a)  :  S \in \mathcal{S}(a)\}$ with $\mathcal S(a)=\{S\in\mathcal N :  a\in\partial V(S)\}$, then, we need to show that $a \in V(N)$.

If $N \in \mathcal{S}(a)$, then it is clear that $a \in V(N)$. We assume that $N \notin \mathcal{S}(a)$. Since $m_N \in \operatorname{co}\{\phi_S(a)  :  S \in \mathcal{S}(a)\}$, by Eq.~\eqref{eq-15}, there exist $\lambda_S^k \ge 0$ for each $S \in \mathcal{S}(a)$ and each $k \in K(S)$ such that $\sum_{S \in \mathcal{S}(a), k \in K(S)} \lambda_S^k \pi^k_S = m_N$, where $K(S) \subset K$ with
\[
K(S) = \{ k \in K  :  u_i(z_{S,i}^k) \ge a_i,\ \forall i \in S \}.
\]
Of course, letting $\lambda_S^k = 0$ for each $S \in \mathcal{S}(a)$ and each $k \in K \setminus K(S)$, we still have $\sum_{S \in \mathcal{S}(a), k \in K} \lambda_S^k \pi^k_S = m_N$.

By Carath\'eodory's theorem, we can extract an affinely independent subcollection of $\{\pi_S^k : S\in\mathcal{S}(a),\, k\in K(S)\}$ that still represents $m_N$ as a convex combination. For all remaining $S\in\mathcal{S}(a)$ and $k\in K$ that are not used in this subcollection, set the corresponding coefficients $\lambda_S^k$ equal to zero. In this way, we obtain nonnegative weight coefficients (denoted still as $\lambda_S^k$ for convenience) defined for every $S\in\mathcal{S}(a)$ and $k\in K$ such that
\(
\sum_{S\in\mathcal{S}(a),\, k\in K} \lambda_S^k \pi_S^k = m_N,
\)
and the vectors $\{\pi_S^k : \lambda_S^k>0\}$ are affinely independent. Then, $\mathcal{S}(a)$ is an independent-$\pi$-balanced family, and we denote it by $\beta'$.
Hence, $\mathcal{S}(a) \in \mathcal{B}$. Denote $\mathcal{S}(a)$ by $\beta'$. We know that $\lambda_S^k$ is a balanced weight (denoted by $\lambda_S^k(\beta')$ to clearly show its relation to $\beta'$) for this independent-$\pi$-balanced family $\beta'$. Given $S \in \beta' \in \mathcal{B}$, it is true that $a \in \partial V(S)$; hence, for each $k \in K(S)$, we have $u_i(z_{S,i}^k) \ge a_i$ for all $i \in S$. Then we have
\begin{equation}\label{eq-5}
u_i(z_{S,i}^{k}) \ge a_i, \quad \forall i \in S \in \beta',\ k \in K(S).
\end{equation}

From the quasi-concavity of $u_i$ and Eqs.~\eqref{eq-13} and \eqref{eq-5}, we have
\begin{equation}\label{eq-4}
\begin{aligned}
u_i(x_i^{\beta'})
&= u_i\biggl( \sum_{S \in \beta'(i), k \in K} n \lambda_S^k(\beta') \pi^{k}_{S,i} z_{S,i}^{k} \biggr) \\
&= u_i\biggl( \sum_{S \in \beta'(i), k \in K(S)} n \lambda_S^k(\beta') \pi^{k}_{S,i} z_{S,i}^{k} \biggr) \\
&\ge \min\{ u_i(z_{S,i}^{k})  :  S \in \beta'(i),\ k \in K(S) \} \ge a_i, \quad \forall i \in N.
\end{aligned}
\end{equation}
Then, $a \in V(N)$ by Eqs.~\eqref{eq-2}, \eqref{eq-3}, and \eqref{eq-4}. Hence, the game $V$ is payoff-dependent balanced with respect to $((\phi_S)_{S \in \mathcal{N}}, m_N)$.

From the above, the game $V$ satisfies all the conditions in Lemma \ref{lem_1}. Then, there exists a core payoff vector $a^* \in V(N) \setminus \bigcup_{S \in \mathcal{N}} \operatorname{int} V(S)$. Of course, $a^* \in V(N)$, so there exists $\beta \in \mathcal{B}$ such that
\begin{equation}\label{eq-6}
u_i(x^\beta_i) \ge a^*_i, \quad \forall i \in N.
\end{equation}
Note that
\begin{equation}\label{eq-7}
\begin{aligned}
\sum_{i \in N} (x^\beta_i - w_i)
&= \sum_{i \in N} \biggl( \sum_{S \in \beta(i), k \in K} n \lambda_S^{k}(\beta) \pi^{k}_{S,i} z_{S,i}^{k} - w_i \biggr) \\
&= \sum_{i \in N} \sum_{S \in \beta(i), k \in K} n \lambda_S^{k}(\beta) \pi^{k}_{S,i} (z_{S,i}^{k} - w_i) \\
&= \sum_{i \in N} \sum_{S \in \beta(i)} \sum_{k \in K} n \lambda_S^{k}(\beta) \pi^{k}_{S,i} (z_{S,i}^{k} - w_i) \\
&= \sum_{S \in \beta} \sum_{i \in S} \sum_{k \in K} n \lambda_S^{k}(\beta) \pi^{k}_{S,i} (z_{S,i}^{k} - w_i) \\
&= \sum_{S \in \beta} \sum_{k \in K} n \lambda_S^{k}(\beta) \sum_{i \in S} \pi^{k}_{S,i} (z_{S,i}^{k} - w_i),
\end{aligned}
\end{equation}
where the second equality follows from Eq.~\eqref{eq-13}. Since $\pi_S^{k}$ and $z_S^{k}$ for all $S \in \mathcal{N}$ and $k \in K$ satisfy Definition \ref{def9}, we have $\sum_{i \in S} \pi^{k}_{S,i} z_{S,i}^{k} = \sum_{i \in S} \pi^{k}_{S,i} w_i$. From Eq.~\eqref{eq-7}, we get $\sum_{i \in N} (x^\beta_i - w_i) = 0$; hence $x^\beta \in F(\mathcal{E})$.

By the implication of the core of $V$, for any $S \in \mathcal{N}$, there is no $a' \in V(S)$ such that $a^*_i < a'_i$ for all $i \in S$. From the initial finite open cover, there exists $k \in K$ such that $x^\beta \in O(x^k) \subset F(\mathcal{E})$. If $x^\beta \in C_f(\mathcal{E})$, the result holds trivially. So we assume $x^\beta \in O(x^k) \setminus C_f(\mathcal{E})$. Since the economy is locally fuzzy-coalition stable, there exists $S' \in \mathcal{N}$ such that
\begin{equation}\label{eq-8}
u_i(z_{S',i}^k) > u_i(x^\beta_i), \quad \forall i \in S'.
\end{equation}
By Eqs.~\eqref{eq-6} and \eqref{eq-8}, it holds that $u_i(z_{S',i}^k) > a^*_i$ for all $i \in S'$. Let $a' = (u_i(z_{S',i}^k), a^*_{-S'})$. Note that $a' \in V(S')$ and $a'_i > a^*_i$ for all $i \in S'$, which contradicts the fact that $a^*$ is in the core of the game $V$. The proof is complete.
\end{proof}

From Proposition \ref{pro_2} and Theorem \ref{thm3-1}, we obtain the following result.

\begin{corollary}\label{cor3-2}
Let $\mathcal{E} = \{X_{i}, w_i, u_i\}_{i \in N}$ be a pure exchange economy satisfying Assumptions \ref{items:A1} and \ref{items:A2}. If all utility functions are upper semi-continuous and quasi-concave, then $C_f(\mathcal{E}) \neq \emptyset$.
\end{corollary}
\begin{remark}
Corollary \ref{cor3-2} is the corresponding result in \citet{liu2022} which generalizes the example in \citet{liu2013} with continuous utility functions.
\end{remark}
If there are only two agents in the economy, the non-emptiness of the fuzzy core can be guaranteed without any quasi-concavity of the utility functions.

\begin{corollary}\label{cor3-1}
Let $\mathcal{E} = \{X_{i}, w_i, u_i\}_{i \in N}$ be a pure exchange economy with $|N| = 2$ satisfying Assumptions \ref{items:A1} and \ref{items:A2}. If the economy $\mathcal{E}$ is locally fuzzy-coalition stable, then $C_f(\mathcal{E}) \neq \emptyset$.
\end{corollary}
\begin{proof}
We check the proof of Theorem \ref{thm3-1} for the case $N = \{1,2\}$. Clearly, $\mathcal{N} = \{\{1\}, \{2\}, \{1,2\}\}$. For each $k \in K$, we have $z_{\{1\}}^k \equiv w_1$ for coalition $\{1\}$ and $z_{\{2\}}^k \equiv w_2$ for coalition $\{2\}$.

We only need to prove that the inequalities in Eq.~\eqref{eq-4} hold without using the quasi-concavity of the utility functions. For the independent-$\pi$-balanced family $\beta'$ in Eq.~\eqref{eq-4}, it holds that $N \notin \beta'$ from the proof of Theorem \ref{thm3-1}. Then, we obtain $\beta' = \{\{1\}, \{2\}\}$. So, by taking a fixed $k_0 \in K(\{i\})$, Eq.~\eqref{eq-4} can be rewritten as
\begin{equation*}
\begin{aligned}
u_i(x_i^{\beta'})
&= u_i\biggl( \sum_{S \in \beta'(i), k \in K} n \lambda_S^k(\beta') \pi^{k}_{S,i} z_{S,i}^{k} \biggr) \\
&= u_i\biggl( w_i \sum_{S \in \beta'(i), k \in K} n \lambda_S^k(\beta') \pi^{k}_{S,i} \biggr) \\
&= u_i(w_i) \\
&= u_i(z_{\{i\}}^{k_0}) \ge a_i, \quad \forall i \in N = \{1,2\}.
\end{aligned}
\end{equation*}
Therefore, Eq.~\eqref{eq-4} holds without using the quasi-concavity of the utility functions. The proof is completed.
\end{proof}

\subsection{TU-fuzzy core and its non-emptiness}

A transferable utility (TU) game is denoted by a characteristic function $v$ from $\mathcal{N}$ to $\mathbb{R}^n$. For each $S \in \mathcal{N}$, the value $v(S)$ denotes the attainable payoff of the crisp coalition $S$. The payoff $v(S)$ can be redistributed arbitrarily among the members of coalition $S$. Employing the spirit of transferable utility, we introduce the following TU-fuzzy core of an exchange economy by Definition \ref{def:weighted-TU-core}.
\begin{definition}\label{def:weighted-TU-core}
Let $\mathcal{E} = \{X_{i}, w_i, u_i\}_{i \in N}$ be a pure exchange economy.
A pair consisting of a commodity allocation and a payoff profile $(x, a) \in F(\mathcal{E}) \times \mathbb{R}^n$ is blocked by $S \in \mathcal{N}$ if there exist $x'_S \in X_S$ and $\pi_S \in \Delta_S$ with $\pi_{S,i} > 0$ for all $i \in S$ and $\sum_{i \in S} \pi_{S,i} x'_{S,i} = \sum_{i \in S} \pi_{S,i} w_i$, such that
\(
\sum_{i \in S} \pi_{S,i} u_i(x'_{S,i}) > \sum_{i \in S} \pi_{S,i} a_i.
\)
A pair $(x, a) \in F(\mathcal{E}) \times \mathbb{R}^n$ with $\sum_{i \in N} u_i(x_i) = \sum_{i \in N} a_i$ is in the TU-fuzzy core of the economy $\mathcal{E}$ if no $S \in \mathcal{N}$ can block the profile.
\end{definition}

For an exchange economy $\mathcal{E} = \{X_{i}, w_i, u_i\}_{i \in N}$, let $C_f^T(\mathcal{E})$ denote the TU-fuzzy core of $\mathcal{E}$. From Definition \ref{def:weighted-TU-core}, no fuzzy coalition $\pi$ with any carrier $S \in \mathcal{N}$ can have a successful threat $x'_S$. Of course, Definition \ref{def:weighted-TU-core} is equivalent to the following Definition \ref{def10-weighted}.

\begin{definition}\label{def10-weighted}
Let $\mathcal{E} = \{X_{i}, w_i, u_i\}_{i \in N}$ be a pure exchange economy.
A pair $(x, a) \in F(\mathcal{E}) \times \mathbb{R}^n$ with $\sum_{i \in N} u_i(x_i) = \sum_{i \in N} a_i$ is in the TU-fuzzy core of the economy $\mathcal{E}$ if there is no $x' \in X$ and no $\pi \in \Delta_N$ with $\sum_{i \in N} \pi_i x'_i = \sum_{i \in N} \pi_i w_i$ such that
\(
\sum_{i \in \operatorname{carr}(\pi)} \pi_i u_i(x'_i) > \sum_{i \in \operatorname{carr}(\pi)} \pi_i a_i.
\)
\end{definition}

For an exchange economy $\mathcal{E} = \{X_{i}, w_i, u_i\}_{i \in N}$ with bounded utility functions, define
\small{
\[
\mathcal{X} = \left\{ (x, a) \in F(\mathcal{E}) \times \mathbb{R}^n \;:\;
\sum_{i \in N} u_i(x_i) = \sum_{i \in N} a_i, \quad
\inf_{x_i \in X_i} u_i(x_i) \leq a_i, \quad \forall i \in N \right\},
\]
}
where bounded utility functions means that $u_i$ is bounded for each $i \in N$. Then $\mathcal{X}$ is a bounded set. Hence, the closure $\mathcal{X}^c$ of $\mathcal{X}$ is compact.
\begin{definition}\label{def:local-weighted-TU}
The economy $\mathcal{E}=\{X_i,w_i,u_i\}_{i\in N}$ is \emph{locally TU-fuzzy-coalition stable} if for each $(x,a)\in\mathcal X^c\setminus C_f^{T}(\mathcal{E})$, there exist
$\pi_S(x,a)\in\Delta_S$ with $\pi_{S,i}(x,a)>0$ for all $i\in S$,
$y_S(x,a)\in X_S$ satisfying
\(
\sum_{i\in S}\pi_{S,i}(x,a)y_{S,i}(x,a)
=
\sum_{i\in S}\pi_{S,i}(x,a)w_i
\)
for every $S\in\mathcal N$, and an open neighborhood $O(x,a)$ of $(x,a)$ in $\mathcal X^c$, such that for each $(x',a')\in O(x,a)\setminus C_f^{T}(\mathcal{E})$, there exists $S\in\mathcal N$ for which
\(
\sum_{i\in S}\pi_{S,i}(x,a)\) \(u_i(y_{S,i}(x,a))
>
\sum_{i\in S}\pi_{S,i}(x,a)a'_i.
\)
\end{definition}

For the proof of main result in this part, we need the following lemma.
\begin{lemma}\label{lem3-2}
Let $T$ be a finite index set. For each $k \in T$, let $p^k \in \Delta_S$ and $c^k \in \mathbb{R}$. Define
\[
A = \bigcup_{k \in T} \left\{ a \in \mathbb{R}^n \ :\ \sum_{i \in S} p_i^k a_i \le c^k \right\},
\]
and define the set-valued mapping $\phi: \partial A \rightrightarrows \Delta_S$ by
\[
\phi(a) = \operatorname{co}\left\{ p^k \in \Delta_S \ :\ \sum_{i \in S} p_i^k a_i = c^k \right\}.
\]
Then, $\phi$ is upper semi-continuous with nonempty compact convex values on $\partial A$.
\end{lemma}

\begin{proof}
For any $a \in \partial A$, since $A$ is a finite union of closed half-spaces, $a$ lies on the boundary of at least one half-space. Thus, there exists $k \in T$ such that
\(
\sum_{i \in S} p_i^k a_i = c^k,
\)
so the set of active indices is nonempty. Consequently, $\phi(a)$ is nonempty. Clearly, it is convex and compact in $\Delta_S$.

Since $\Delta_S$ is compact, it suffices to show that the graph of $\phi$ is closed.
Let $\{a^m\} \subset \partial A$ with $a^m \to a^0 \in \partial A$. Take $y^m \in \phi(a^m)$ with $y^m \to y^0$. It needs to show that $y^0 \in \phi(a^0)$.
For each $m$, by the definition of $\phi(a^m)$, the point $y^m$ can be written as a convex combination of those $p^k$ that are active at $a^m$. Thus, there exist coefficients $\lambda_k^m \ge 0$ for $k \in T$ such that
\[
\sum_{k \in T} \lambda_k^m = 1,
\qquad
y^m = \sum_{k \in T} \lambda_k^m p^k,
\]
and
\begin{equation}\label{implication}
\lambda_k^m > 0 \quad \Longrightarrow \quad \sum_{i \in S} p_i^k a_i^m = c^k.
\end{equation}
Since $T$ is finite, the simplex
\(
\Delta_T = \left\{ \lambda \in \mathbb{R}^{|T|}_+  :  \sum_{k \in T} \lambda_k = 1 \right\}
\)
is compact. Hence the sequence $\{\lambda^m\}$ admits a convergent subsequence, which we still denote by $\{\lambda^m\}$, with limit $\lambda^0 \in \Delta_T$. By continuity of the finite sum,
\[
y^0 = \lim_{m \to \infty} y^m
= \lim_{m \to \infty} \sum_{k \in T} \lambda_k^m p^k
= \sum_{k \in T} \lambda_k^0 p^k.
\]

Now suppose $\lambda_k^0 > 0$ for some $k \in T$. Then, for all sufficiently large $m$, we have $\lambda_k^m > 0$, so by the implication \eqref{implication} above,
\(
\sum_{i \in S} p_i^k a_i^m = c^k.
\)
Taking the limit as $m \to \infty$ and using $a_i^m \to a_i^0$, we obtain
\(
\sum_{i \in S} p_i^k a_i^0 = c^k.
\)
Therefore, every index $k$ with positive coefficient $\lambda_k^0$ is active at $a^0$. Hence
\[
y^0 = \sum_{k \in T} \lambda_k^0 p^k
= \sum_{\{k\in T : \lambda_k^0 > 0\}} \lambda_k^0 p^k
\in \operatorname{co}\left\{ p^k\in \Delta_S : \sum_{i \in S} p_i^k a_i^0 = c^k \right\}
= \phi(a^0).
\]
This proves that the graph of $\phi$ is closed. Therefore, $\phi$ is upper semi-continuous.
\end{proof}

By constructing an NTU game that is payoff-dependent balanced with respect to a special transfer rate rule, the second main result (Theorem \ref{thm3-2}) shows the non-emptiness of the TU-fuzzy core for a locally TU-fuzzy-coalition-stable economy with concave utility functions.
\begin{theorem}\label{thm3-2}
Let $\mathcal{E}=\{X_i,w_i,u_i\}_{i\in N}$ be a pure exchange economy satisfying local TU-fuzzy-coalition stability and Assumptions \ref{items:A1} and \ref{items:A2}. Assume that for each $i\in N$, the utility function $u_i$ is concave. Then, the TU-fuzzy core $C_f^{T}(\mathcal{E})$ is nonempty.
\end{theorem}

\begin{proof}
Suppose, by contradiction, that $C_f^{T}(\mathcal{E})=\emptyset$. Then, for every $(x,a)\in\mathcal X^c$, the local TU-fuzzy-coalition stability condition provides data $\pi_S(x,a)$, $y_S(x,a)$, and an open neighborhood $O(x,a)$ as in Definition~\ref{def:local-weighted-TU}. Since $\mathcal{X}^c$ is compact, there exists a finite set $K$ such that $\bigcup_{k \in K} O(x^k, a^k)$ is a subcover of the open cover $\bigcup_{(x,a) \in \mathcal{X}^c} O(x,a)$ of $\mathcal{X}^c$.  For each $k\in K$ and $S\in\mathcal N$, write
\[
\pi_S^k=\pi_S(x^k,a^k),\qquad y_S^k=y_S(x^k,a^k).
\]
Let $\mathcal B$ be the collection of all independent-$\pi$-balanced families of $\mathcal N$ with respect to $\{\pi_S^k\}_{S\in\mathcal N,k\in K}$.

For each $S\in\mathcal N$ with $S\ne N$, define the feasible payoff set by
\begin{equation}\label{eq:VSw}
V(S)
=
\bigcup_{k\in K}
\left\{
a\in\mathbb R^n
\;:\;
\sum_{i\in S}\pi_{S,i}^k a_i
\le
\sum_{i\in S}\pi_{S,i}^k u_i(y_{S,i}^k)
\right\}.
\end{equation}
For the grand coalition, set
\begin{equation}\label{eq:VNw}
V(N)
=
\left\{
a\in\mathbb R^n
\;:\;
\sum_{i\in N}a_i
\le
v(N)
\right\},
\end{equation}
where
\begin{equation}\label{eq:vNw}
v(N)
=
\max_{\beta\in\mathcal B}
\sum_{i\in N}u_i(x_i^\beta),
\end{equation}
and
\begin{equation}\label{eq:xbeta-w}
x_i^\beta
=
\sum_{S\in\beta(i),k\in K}
n\lambda_S^k(\beta)\pi_{S,i}^k y_{S,i}^k,
\end{equation}
with $\lambda_S^k(\beta)$ ranging over all finite balanced weights associated with the independent -$\pi$-balanced family $\beta$.

The set $V(S)$ is clearly nonempty and comprehensive for each $S\in \mathcal{N}$, and is closed for each $S\in \mathcal{N}\backslash N$. From the finiteness of $\mathcal{B}$ and the finiteness of balanced weights for each $\beta\in \mathcal{B}$, $V(N)$ is also closed; and each $V(S)$ is bounded above. Therefore, the game satisfies Assumptions \ref{item:B1}--\ref{item:B3}.

For each $S\ne N$, define the transfer rate rule $\phi_S:\partial V(S)\rightrightarrows\Delta_S$ by
\begin{equation}\label{eq:phiSw}
\phi_S(a)
=
\operatorname{co}
\left\{
\pi_S^k\in\Delta_S
\;:\;
k\in K,\
\sum_{i\in S}\pi_{S,i}^k a_i
=
\sum_{i\in S}\pi_{S,i}^k u_i(y_{S,i}^k)
\right\},
\end{equation}
for each $a\in \partial V(S)$. Obviously, each $\phi_S$ has nonempty compact convex values.
 Furthermore, by Lemma~\ref{lem3-2}, if we take $T=K$, $p^k=\pi_S^k$, and $c^k=\sum_{i \in S}\pi_{S,i}^k u_i(y_{S,i}^k)$, it obtains the upper semi-continuity of $\phi_S$.

We now verify that the game $V$ is payoff-dependent balanced with respect to $((\phi_S)_{S\in\mathcal N}$, $m_N)$. Let $a\in\partial W$, where $W=\bigcup_{S\in\mathcal N}V(S)$, and suppose that
\[
m_N\in\operatorname{co}\{\phi_S(a) :  S\in\mathcal S(a)\},
\]
with $\mathcal S(a)=\{S\in\mathcal N :  a\in\partial V(S)\}$. If $N\in\mathcal S(a)$, then immediately $a\in V(N)$. Hence assume $N\notin\mathcal S(a)$.

Similar with the corresponding proof in Theorem \ref{thm3-1}, there exists an independent-$\pi$-balanced family $\beta'=\mathcal S(a)\in \mathcal{B}$ and balanced weight coefficients $\lambda_S^k(\beta')\ge 0$ with respect to $\beta'$ such that
\[
m_N
=
\sum_{S\in\beta'}\sum_{k\in K(S)}
\lambda_S^k(\beta')\pi_S^k,
\]
and other $\lambda_S^k(\beta')= 0$ for each $S\in\beta', k\not\in K(S)$,
where
\[
K(S)
=
\left\{
k\in K
\;:\;
\sum_{i\in S}\pi_{S,i}^k a_i
=
\sum_{i\in S}\pi_{S,i}^k u_i(y_{S,i}^k)
\right\}.
\]
Thus, for every $i\in N$,
\begin{equation}\label{eq:bal-w}
\sum_{S\in\beta'(i),k\in K(S)}
n\lambda_S^k(\beta')\pi_{S,i}^k
=
1.
\end{equation}
Note that for each $S\in\beta'$ and $k\in K(S)$, we have
\begin{equation}\label{eq:active-w}
\sum_{i\in S}\pi_{S,i}^k u_i(y_{S,i}^k)
=
\sum_{i\in S}\pi_{S,i}^k a_i.
\end{equation}

Using the concavity of each $u_i$ and the balanced weights, we obtain
\begin{equation}\label{vN21}
\begin{aligned}
v(N)
&\ge
\sum_{i\in N}u_i(x_i^{\beta'})\\
&= \sum_{i \in N} u_i\biggl( \sum_{S \in \beta'(i), k \in K} n \lambda^k_S(\beta') \pi_{S,i}^{k} y_{S,i}^{k} \biggr) \\
&=
\sum_{i\in N}
u_i\left(
\sum_{S\in\beta'(i),k\in K(S)}
n\lambda_S^k(\beta')\pi_{S,i}^k y_{S,i}^k
\right)\\
&\ge
\sum_{i\in N}
\sum_{S\in\beta'(i),k\in K(S)}
n\lambda_S^k(\beta')\pi_{S,i}^k u_i(y_{S,i}^k)\\
&=
\sum_{S\in\beta'}\sum_{k\in K(S)}
n\lambda_S^k(\beta')
\sum_{i\in S}
\pi_{S,i}^k u_i(y_{S,i}^k).
\end{aligned}
\end{equation}
By \eqref{eq:active-w}, the last expression in \eqref{vN21} equals
\[
\sum_{S\in\beta'}\sum_{k\in K(S)}
n\lambda_S^k(\beta')
\sum_{i\in S}
\pi_{S,i}^k a_i
=
\sum_{i\in N}
a_i
\left(
\sum_{S\in\beta'(i),k\in K(S)}
n\lambda_S^k(\beta')\pi_{S,i}^k
\right).
\]
Using \eqref{eq:bal-w}, we have that
\begin{equation}\label{vNinquality}
v(N)\ge \sum_{i\in N}a_i,
\end{equation}
so $a\in V(N)$ by \eqref{eq:VNw}. Therefore, the game $V$ is payoff-dependent balanced.

By Lemma~\ref{lem_1}, there exists a core payoff vector
\[
\bar{a}\in
V(N)\setminus \bigcup_{S\in\mathcal N}\operatorname{int}V(S).
\]
From the construction of $V(N)$, it follows that
\(
\sum_{i\in N}\bar a_i
=
v(N),
\)
and there exist $\beta\in\mathcal B$ and associated balanced weights such that
\begin{equation}\label{eq-core}
\sum_{i\in N}\bar a_i
=
v(N)
=
\sum_{i\in N}u_i(x_i^\beta).
\end{equation}

As in the proof of Theorem~\ref{thm3-1}, one verifies that $x^\beta\in F(\mathcal{E})$. Moreover, we assert that
\begin{equation}\label{eq-inf}
\inf_{x_i \in X_i} u_i(x_i) \leq \bar{a}_i, \quad \forall i \in N.
\end{equation}
In fact, for each singleton coalition $S=\{i\}$, since $\pi_{\{i\},i}^k=1$, the feasible set reduces to
\[
V(\{i\})
=
\left\{a\in\mathbb R^n \;:\; a_i\le \max_{k\in K}u_i(y_{\{i\},i}^k)\right\}.
\]
Because $\bar a$ is a core payoff vector, $\bar a\notin \operatorname{int}V(\{i\})$, which implies
\[
\bar a_i
\ge
\max_{k\in K}u_i(y_{\{i\},i}^k)
\ge
\inf_{z_i\in X_i}u_i(z_i).
\]
As $i\in N$ is arbitrary, this yields \eqref{eq-inf}. Combining this with \eqref{eq-core}, we conclude that $(x^\beta,\bar a)\in\mathcal X$.

By the finite open cover, there exists $k\in K$ such that
\(
(x^\beta,\bar a)\in O(x^k,a^k).
\)
If $(x^\beta,\bar a)\in C_f^{T}(\mathcal{E})$, the proof is complete. Otherwise, by local TU-fuzzy-coalition stability, there exists $S\in\mathcal N$ such that
\[
\sum_{i\in S}\pi_{S,i}^k u_i(y_{S,i}^k)
>
\sum_{i\in S}\pi_{S,i}^k \bar a_i.
\]
According to \eqref{eq:VSw}, this strict inequality implies that
\[
\bar a
\in
\operatorname{int}
\left\{
a\in\mathbb R^n
\;:\;
\sum_{i\in S}\pi_{S,i}^k a_i
\le
\sum_{i\in S}\pi_{S,i}^k u_i(y_{S,i}^k)
\right\}
\subset
\operatorname{int}V(S),
\]
contradicting the fact that $\bar a$ is a core payoff vector. Therefore $C_f^{T}(\mathcal{E})\ne\emptyset$.
\end{proof}

 When there are only two agents, the TU-fuzzy core is always nonempty without any concavity of the utility functions.
\begin{corollary}\label{cor3-3}
Let $\mathcal{E} = \{X_{i}, w_i, u_i\}_{i \in N}$ be a pure exchange economy with $|N| = 2$ satisfying Assumptions \ref{items:A1} and \ref{items:A2}. If $\mathcal{E}$ is locally TU-fuzzy-coalition stable, then $C^T_f(\mathcal{E}) \neq \emptyset$.
\end{corollary}
\begin{proof}
We check the proof of Theorem \ref{thm3-2} for the case $N = \{1,2\}$. Obviously, $\mathcal{N} = \{\{1\}, \{2\}, \{1,2\}\}$. For each $k \in K$, for coalition $\{i\}$, we have $y_{\{i\},i}^k \equiv w_i$ and $\pi_{\{i\},i}^{k} = 1 \in \Delta_{\{i\}} = \{1\}$.

What we need to prove is that Eqs.~\eqref{eq:bal-w} and \eqref{eq:active-w} imply \eqref{vNinquality}, $v(N) \geq \sum_{i \in N} a_i\,(a\in V(N))$, without invoking the concavity of the utility functions. Since $N \notin \beta'$ for the independent-$\pi$-balanced family $\beta'$ in Eq.~\eqref{eq:bal-w}, we must have $\beta' = \{\{1\}, \{2\}\}$. For each $S = \{i\} \in \beta'$, we can rewrite \eqref{eq:active-w},
\[
\sum_{i\in S}\pi_{S,i}^k u_i(y_{S,i}^k)
=
\sum_{i\in S}\pi_{S,i}^k a_i,
\]
 as
 \begin{equation}\label{eq-30}
 u_i(w_i)
= a_i.
\end{equation}
Furthermore, Eq.~\eqref{vN21} is rewritten as
\begin{equation}\label{eq-31}
\begin{aligned}
v(N) &\geq \sum_{i \in N} u_i(x^{\beta'}_{i}) \\
&= \sum_{i \in N} u_i\biggl( \sum_{S \in \beta'(i), k \in K} n \lambda^k_S(\beta') \pi_{S,i}^{k} y_{S,i}^{k} \biggr) \\
&= \sum_{i \in N} u_i\biggl( \sum_{S \in \beta'(i), k \in K(S)} n \lambda^k_S(\beta') \pi_{S,i}^{k} y_{S,i}^{k} \biggr) \\
&= \sum_{i \in N} u_i\biggl( w_i \sum_{S \in \beta'(i), k \in K(S)} n \lambda^k_S(\beta') \pi_{S,i}^{k} \biggr) \\
&= \sum_{i \in N} u_i(w_i).
\end{aligned}
\end{equation}
From Eqs.~\eqref{eq-30} and \eqref{eq-31}, we have $v(N) \geq \sum_{i \in N} a_i$. The proof is completed.
\end{proof}

Without the concavity of the utility functions, Corollary \ref{cor3-3} shows that if an economy with two agents satisfies local TU-fuzzy-coalition stability, then it will possess a nonempty TU-fuzzy core. A question is whether the concavity of the utility functions in Theorem \ref{thm3-2} is necessary.

The following Example \ref{exam-3} shows the complexity of the (TU-)fuzzy core. The example, with three agents and quasi-concave (but not concave) utility functions, satisfies local TU-fuzzy-coalition stability, and it still has a nonempty TU-fuzzy core. Thus, Theorem \ref{thm3-2} is a sufficient condition for the non-emptiness of the TU-fuzzy core of an economy. In addition, the example does not satisfy local fuzzy-coalition stability, and yet it still has a nonempty fuzzy core.

\begin{example}\label{exam-3}
Let $L = \mathbb{R}^3$ and $N = \{1,2,3\}$. For each $i \in N$, the consumption set is
\[
X_i = [0,2]^3 \subset \mathbb{R}^3 = L.
\]
The initial endowments are
\(
w_1 = (1,0,0),\qquad w_2 = (0,1,0),\qquad w_3 = (0,0,1).
\)

For each $i \in N$, define $u_i : X_i \to \mathbb{R}$ as follows:
\[
u_1(x_{1}) =
\begin{cases}
3, & x_{1,1} + x_{1,2} + x_{1,3} = 1 \text{ and } x_{1,2} + x_{1,3} \le 0.5, \\
0, & \text{otherwise},
\end{cases}
\]
\[
u_2(x_{2}) =
\begin{cases}
1, & x_{2,1} + x_{2,2} + x_{2,3} = 1 \text{ and } x_{2,1} + x_{2,3} < 0.75, \\
0, & \text{otherwise},
\end{cases}
\]
and
\[
u_3(x_{3}) =
\begin{cases}
1, & x_{3,1} + x_{3,2} + x_{3,3} = 1 \text{ and } x_{3,1} + x_{3,2} \le 0.5, \\
0, & \text{otherwise}.
\end{cases}
\]

Then, an economy $\mathcal{E} = \{X_{i}, w_i, u_i\}_{i \in N}$ is defined. The feasible set is
\small{
\[
F(\mathcal E) = \left\{ (x_1,x_2,x_3) \in X_1 \times X_2 \times X_3 : x_1 + x_2 + x_3 = w_1 + w_2 + w_3 = (1,1,1) \right\}.
\]
}
Clearly,
\small{
\[
\mathcal{X} = \left\{ (x, a) \in F(\mathcal{E}) \times \mathbb{R}^n \;:\;
\sum_{i \in N} u_i(x_i) = \sum_{i \in N} a_i, \quad
\inf_{x_i \in X_i} u_i(x_i) = 0 \leq a_i, \quad \forall i \in N \right\}
\]
}
is bounded and $\mathcal X^c$ is compact. All utility functions are discontinuous.

For each $i \in N$, let
\begin{align*}
H_1 &= \{ (x,y,z) \in [0,2]^3  :  x + y + z = 1,\ y + z \le 0.5 \}, \\
H_2 &= \{ (x,y,z) \in [0,2]^3  :  x + y + z = 1,\ x + z < 0.75 \}, \\
H_3 &= \{ (x,y,z) \in [0,2]^3  :  x + y + z = 1,\ x + y \le 0.5 \}.
\end{align*}
Then, the {\bf fuzzy core} of the economy is
\begin{equation}\label{eq-32}
C_f(\mathcal E) = \{ (x_1,x_2,x_3) \in F(\mathcal E)  :  x_1 \in H_1,\ x_2 \in H_2,\ x_3 \in H_3 \}.
\end{equation}
Indeed, if $x \in F(\mathcal E)$ and every agent $i$ attains its maximum utility $M_i$, then no $S \in \mathcal{N}$ can strictly improve all its members, so $x \in C_f(\mathcal E)$. Conversely, if $x \in F(\mathcal E)$ and some agent $i$ does not attain its maximum value $M_i$, then the singleton coalition $\{i\}$ can strictly improve agent $i$ by the threat $y_i = w_i$, since $u_i(w_i) = M_i > u_i(x_i)$. Hence $x \notin C_f(\mathcal E)$.

It can be checked that the {\bf TU-fuzzy core} of the economy is nonempty, and
\begin{equation}\label{eq-33}
C^T_f(\mathcal E) = \{ (x,a) \in F(\mathcal{E}) \times \mathbb{R}^n  :  x \in C_f(\mathcal E),\ a = (3,1,1) \}.
\end{equation}
Indeed, first suppose that $x \in C_f(\mathcal E)$ and $a=(3,1,1)$. Then $\sum_{i\in N}u_i(x_i)=5=\sum_{i\in N}a_i$. We show that $(x,a)$ cannot be blocked by any fuzzy coalition. If $S=\{i\}$ is a singleton, then $\pi_{\{i\}}=1$, and the best threat $y_{\{i\}}=w_i$ gives utility $u_i(w_i)$ equal to $3$ for $i=1$ and $1$ for $i=2,3$. Hence the weighted payoff of the threat, $\pi_{\{i\}}u_i(y_{\{i\}})=u_i(w_i)$, is exactly $a_i$, not strictly greater. For any coalition $S$ with $|S|\ge 2$, any convex combination for any feasible threats with positive weights, $\sum_{i\in S}\pi_{S,i}u_i(y_i)$, cannot strictly dominate $\sum_{i\in S}\pi_{S,i}a_i$ since $a_i$ equals the maximum utility for every $i$. Hence, $(x,a)$ is in the TU-fuzzy core.

Conversely, suppose $(x,a)\in C_f^T(\mathcal E)$. Then $\sum_{i\in N}u_i(x_i)=\sum_{i\in N}a_i$. Note that $\max_{x\in X}$ $\sum_{i\in N}u_i(x_i)=5$. If $a_1<3$, then the singleton coalition $\{1\}$ has a threat $y_{\{1\}}=w_1$ with $\pi_{\{1\}}=1$ such that
\[
\pi_{\{1\}}u_1(y_{\{1\}})=3>a_1=\pi_{\{1\}}a_1,
\]
so $(x,a)$ is blocked. Hence, we must have $a_1\ge 3$. Similarly, $a_2\ge 1$ and $a_3\ge 1$, because otherwise the corresponding singleton coalition would block. Since $\sum_{i\in N}a_i=\sum_{i\in N}u_i(x_i)\le 5$, we must have $a_1=3$, $a_2=1$, $a_3=1$, and $\sum_{i\in N}u_i(x_i)=5$. Therefore each $u_i(x_i)$ equals its maximum value, so $x\in C_f(\mathcal E)$. This proves \eqref{eq-33}.

Next, we show that the economy $\mathcal{E}$ does not satisfy local fuzzy-coalition stability but does satisfy local TU-fuzzy-coalition stability, by the following two propositions.

\begin{proposition}
Local fuzzy-coalition stability fails for the economy $\mathcal{E}$.
\end{proposition}
\begin{proof}
Let
\[
x^* = \big( (0.5,0.25,0.25),\ (0.5,0.25,0.25),\ (0,0.5,0.5) \big).
\]
We have $x^* \in F(\mathcal E)$ since $x^*_1 + x^*_2 + x^*_3 = (1,1,1)$. Note that
\[
u_1(x^*_1) = 3,\qquad u_2(x^*_2) = 0,\qquad u_3(x^*_3) = 1.
\]
Thus, $x^*_1 \in H_1$, $x^*_2 \notin H_2$, and $x^*_3 \in H_3$. Then $x^* \notin C_f(\mathcal E)$. For this allocation $x^* \notin C_f(\mathcal E)$, agent 1 attains the maximum utility $3$, so any $S$ containing agent 1 cannot strictly improve agent 1. Agent 3 attains the maximum utility $1$, so any $S$ containing agent 3 cannot strictly improve agent 3. The only $S$ containing neither agent 1 nor agent 3 is $S = \{2\}$.

For the singleton coalition $S = \{2\}$, the threat $z_2$ in Definition \ref{def9} must be $w_2 = (0,1,0)$, yielding threat utility $u_2(z_2) = 1$.

However, $x^*_2 = (0.5,0.25,0.25)$ satisfies $x_{2,1} + x_{2,3} = 0.75$, which is the boundary of the open set $H_2$. Therefore, in every open neighborhood $O(x^*)$ of $x^*$, there exists $y \in O(x^*)$ such that $y_2 \in H_2$, i.e., $u_2(y_2) = 1$. For such $y$, the threat condition $u_2(z_2) = 1 > u_2(y_2)$ fails.

Thus, no $S \in \mathcal{N}$, threat $z_S$, and neighborhood $O(x^*)$ can strictly improve all members of $S$ simultaneously on $O(x^*)$. Hence, the economy does not satisfy local fuzzy-coalition stability.
\end{proof}

\begin{proposition}
The economy $\mathcal{E}$ satisfies the local TU-fuzzy-coalition stability in Definition \ref{def:local-weighted-TU}.
\end{proposition}
\begin{proof}
Take any $(x,a) \notin C^T_f(\mathcal E)$. Take an arbitrary open neighborhood $O(x,a)$ of $(x,a)$ in $\mathcal{X}^c$. For every $(x',a') \in O(x,a) \setminus C^T_f(\mathcal E)$, we have $x' \notin C_f(\mathcal E)$ or $a' \neq (3,1,1)$. Note that by the definition of $\mathcal X^c$, it holds that $a'_1 \geq 0$, $a'_2 \geq 0$, and $a'_3 \geq 0$.

If $x' \notin C_f(\mathcal E)$, then the three agents cannot all attain their maximum utilities (recall that they are $3$, $1$, and $1$) simultaneously, so
\(
\sum_{i=1}^3 u_i(x'_i) \le 4.
\)
By the definition of $\mathcal X^c$, $\sum_{i=1}^3 a'_i = \sum_{i=1}^3 u_i(x'_i)$, hence
\[
\sum_{i=1}^3 a'_i \le 4.
\]
If $a'_1 < 3$, then the singleton coalition $\{1\}$ has a threat $y_{\{1\}} = w_1$ with $\pi_{\{1\}} = 1$ by Definition \ref{def:local-weighted-TU} such that $\pi_{\{1\}} u_1(y_{\{1\}}) = 3 > a'_1 = \pi_{\{1\}} a'_1$. If $a'_1 = 3$, then from $\sum_{i=1}^3 a'_i \le 4$, we have $a'_2 + a'_3 \le 1$. Hence, there exists $i \in \{2,3\}$ such that $a'_i < 1$, and the corresponding singleton coalition $\{i\}$ has a threat $y_{\{i\}} = w_i$ with $\pi_{\{i\}} = 1$ such that $\pi_{\{i\}} u_i(y_{\{i\}}) = 1 > a'_i = \pi_{\{i\}} a'_i$.

If $a' \neq (3,1,1)$, by the definition of $\mathcal X^c$, $\sum_{i=1}^3 u_i(x'_i) = \sum_{i=1}^3 a'_i$. Since $\max_{x \in X}$ $\sum_{i=1}^3 u_i(x_i)$ $= 5$, we have $\sum_{i=1}^3 a'_i \leq 5$.

If $a'_1 < 3$, then the singleton coalition $\{1\}$ has a threat $y_{\{1\}} = w_1$ with $\pi_{\{1\}} = 1$ such that $\pi_{\{1\}} u_1(y_{\{1\}}) = 3 > a'_1 = \pi_{\{1\}} a'_1$.

If $a'_1 > 3$, then we have $a'_2 + a'_3 < 2$. Then $a'_2 < 1$ or $a'_3 < 1$ holds. Without loss of generality, let $a'_2 < 1$. Then, the singleton coalition $\{2\}$ has a threat $y_{\{2\}} = w_2$ with $\pi_{\{2\}} = 1$ such that $\pi_{\{2\}} u_2(y_{\{2\}}) = 1 > a'_2 = \pi_{\{2\}} a'_2$.

If $a'_1 = 3$, then we have $a'_2 + a'_3 \leq 2$ and $a'_2 \neq 1$ or $a'_3 \neq 1$ (since $a' \neq (3,1,1)$). Without loss of generality, let $a'_2 < 1$. Again, the singleton coalition $\{2\}$ has a threat $y_{\{2\}} = w_2$ with $\pi_{\{2\}} = 1$ such that $\pi_{\{2\}} u_2(y_{\{2\}}) = 1 > a'_2 = \pi_{\{2\}} a'_2$.

Then, for every $(x',a') \in O(x,a) \setminus C^T_f(\mathcal E)$, there exists at least one singleton coalition, say $\{i\}$, which has a threat $y_{\{i\}} = w_i$ with $\pi_{\{i\}} = 1$ such that $\pi_{\{i\}} u_i(y_{\{i\}}) > \pi_{\{i\}} a'_i$. Therefore, the economy is locally TU-fuzzy-coalition stable.
\end{proof}
\end{example}

\section{Concluding remarks}

In this paper, we have provided a further investigation into the non-emptiness of the fuzzy core and, furthermore, introduced the notion of the TU-fuzzy core for pure exchange economies. To the best of our knowledge, this is the first attempt to define a transferable-utility version of the fuzzy core in the context of exchange economies, thereby establishing a link between the fuzzy-core literature and the classical TU-game approach.
Our analysis contributes to the existing literature in several respects.

First, we introduced the notion of an independent-$\pi$-balanced family of subsets of the agent set, which generalizes the classical balancedness conditions. Based on this notion, we constructed two payoff-dependent NTU games that serve as bridges between the fuzzy core (TU-fuzzy core) of an exchange economy and the core of NTU games.

Second, we proposed two stability conditions for exchange economies: local fuzzy-coalition stability and local TU-fuzzy-coalition stability. These conditions are inspired mainly by the coalitional $C$-security and coalitional deviation properties developed in the literature on discontinuous games. Under these stability conditions, we established the non-emptiness of the fuzzy core (Theorem \ref{thm3-1}) and the non-emptiness of the TU-fuzzy core (Theorem \ref{thm3-2}) for economies with quasi-concave and concave utility functions, respectively. Notably, when the economy consists of only two agents, the quasi-concavity assumption for the fuzzy core and the concavity assumption for the TU-fuzzy core can be dispensed with entirely (Corollaries \ref{cor3-1} and \ref{cor3-3}).

Third, we constructed several examples to illustrate the scope and limitations of our results. Example \ref{exam-1} demonstrates that an economy with discontinuous utility functions (neither upper semi-continuous nor lower semi-continuous) may still possess a nonempty fuzzy core. The fuzzy core in the example is nonempty, which follows from Theorem \ref{thm3-1}, as all the prerequisites of the theorem are satisfied. Example \ref{exam-3} further reveals the complexity of the (TU-)fuzzy core: with three agents and quasi-concave (but not concave) utility functions, the economy satisfies local TU-fuzzy-coalition stability and has a nonempty TU-fuzzy core, yet it fails to satisfy local fuzzy-coalition stability while still having a nonempty fuzzy core. These examples indicate that the sufficient conditions we propose are by no means necessary, and there remains a substantial gap between the sufficient conditions and the actual existence of fuzzy cores for discontinuous economies.

Overall, the results presented in this paper contribute to a better understanding of the fuzzy core as a cooperative solution concept for economies with discontinuous fundamentals, and introduce the TU-fuzzy core as a new concept. We hope that the concepts developed here will be useful in future studies of cooperative behavior in fuzzy and discontinuous environments.

%



\begin{thebibliography}{}

\bibitem[Allouch and Predtetchinski, 2008]{allouch2008}
Allouch, N. and Predtetchinski, A. (2008).
\newblock On the non-emptiness of the fuzzy core.
\newblock {\em International Journal of Game Theory}, 37(2):203--210.

\bibitem[Aubin, 1979]{aubin1979}
Aubin, J.-P. (1979).
\newblock {\em Mathematical Methods of Game and Economic Theory}.
\newblock North-Holland, Amsterdam.

\bibitem[Aumann, 1961]{aumann1961}
Aumann, R. (1961).
\newblock The core of a cooperative game without side payments.
\newblock {\em Transactions of the American Mathematical Society},
  98(3):539--552.

\bibitem[Barelli and Meneghel, 2013]{barelli2013}
Barelli, P. and Meneghel, I. (2013).
\newblock A note on the equilibrium existence problem in discontinuous games.
\newblock {\em Econometrica}, 81(2):813--824.

\bibitem[Basile and Scalzo, 2020]{scalzo2020ijgt}
Basile, A. and Scalzo, V. (2020).
\newblock Non-emptiness of the alpha-core: sufficient and necessary conditions.
\newblock {\em International Journal of Game Theory}, 49(4):1143--1153.

\bibitem[Bonnisseau and Iehl{\'e}, 2003]{bonnisseau2003payoff}
Bonnisseau, J.-M. and Iehl{\'e}, V. (2003).
\newblock Payoff-dependent balancedness and cores.
\newblock {\em Cahiers dela MSE S{\'e}rie Bleue \#2003.45 Universite Paris-I}.

\bibitem[Bonnisseau and Iehl{\'e}, 2007]{bonnisseau2007payoff}
Bonnisseau, J.-M. and Iehl{\'e}, V. (2007).
\newblock Payoff-dependent balancedness and cores.
\newblock {\em Games and Economic Behavior}, 61(1):1--26.

\bibitem[Chen et~al., 2024]{chen2024existence}
Chen, T., Chen, K., and Zhang, Y. (2024).
\newblock Existence and continuity theorems of $\alpha$-core of
  multi-leader-follower games with set payoffs.
\newblock {\em Journal of Computational and Applied Mathematics}, 439:115610.

\bibitem[Chi and Song, 2026]{Chi2026}
Chi, X.-Y. and Song, Q.-Q. (2026).
\newblock The existence and hadamard well-posedness of cooperative equilibria
  for discontinuous socially structured games.
\newblock {\em Journal of Mathematical Analysis and Applications},
  558(2):130430.

\bibitem[Debreu and Scarf, 1963]{debreu1963}
Debreu, G. and Scarf, H. (1963).
\newblock A limit theorem on the core of an economy.
\newblock {\em International Economic Review}, 4(3):235--246.

\bibitem[Edgeworth, 1881]{edgeworth1881}
Edgeworth, F.~Y. (1881).
\newblock {\em Mathematical Psychics: An Essay on the Application of
  Mathematics to the Moral Sciences}.
\newblock C. Kegan Paul \& Co., London.

\bibitem[Florenzano, 2003]{florenzano2003}
Florenzano, M. (2003).
\newblock {\em General Equilibrium Analysis: Existence and Optimality
  Properties of Equilibria}.
\newblock Kluwer, Boston.

\bibitem[Herv{\'e}s-Beloso et~al., 2026]{herves2026}
Herv{\'e}s-Beloso, C., Herv{\'e}s-Est{\'e}vez, J., and Salas-Guti{\'e}rrez, A.
  (2026).
\newblock On {Aubin}'s veto mechanism: {A} survey.
\newblock {\em Journal of Public Economic Theory}, 28(4):e70127.

\bibitem[Li and Jia, 2023]{jia2023fuzzy}
Li, Y. and Jia, W. (2023).
\newblock Existence and well-posedness of the $\alpha$-core for generalized
  fuzzy games.
\newblock {\em Fuzzy Sets and Systems}, 458:108--117.

\bibitem[Liu, 2022]{liu2022}
Liu, J. (2022).
\newblock Non-emptiness of the fuzzy core in a finite production economy with
  infinite-dimensional commodity space.
\newblock {\em Fuzzy Sets and Systems}, 447:123--136.

\bibitem[Liu, 2023]{liu2023}
Liu, J. (2023).
\newblock Necessary and sufficient conditions for solutions of the {Ky Fan}
  minimax inequality and the non-emptiness of fuzzy cores in economies.
\newblock {\em Fuzzy Sets and Systems}, 467:108580.

\bibitem[Liu and Liu, 2013]{liu2013}
Liu, J. and Liu, X. (2013).
\newblock A necessary and sufficient condition for an {NTU} fuzzy game to have
  a non-empty fuzzy core.
\newblock {\em Journal of Mathematical Economics}, 49(2):150--156.

\bibitem[Marakulin, 2022]{marakulin2022}
Marakulin, V. (2022).
\newblock On the existence of a fuzzy core in an exchange economy.
\newblock In {\em Lecture Notes in Computer Science}. Springer.

\bibitem[McLennan et~al., 2011]{mclennan2011}
McLennan, A., Monteiro, P., and Tourky, R. (2011).
\newblock Games with discontinuous payoffs: a strengthening of {Reny}'s
  existence theorem.
\newblock {\em Econometrica}, 79(5):1643--1664.

\bibitem[Nessah and Tian, 2016]{NessahTian2016}
Nessah, R. and Tian, G. (2016).
\newblock On the existence of {N}ash equilibrium in discontinuous games.
\newblock {\em Economic Theory}, 61(3):515--540.

\bibitem[Predtetchinski, 2005]{predtetchinski2005fuzzy}
Predtetchinski, A. (2005).
\newblock The fuzzy core and the-balanced core.
\newblock {\em Economic Theory}, 26(3):717--724.

\bibitem[Predtetchinski and Herings, 2004]{Hering2004necessary}
Predtetchinski, A. and Herings, P. J.-J. (2004).
\newblock A necessary and sufficient condition for non-emptiness of the core of
  a non-transferable utility game.
\newblock {\em Journal of Economic Theory}, 116(1):84--92.

\bibitem[Prokopovych, 2011]{Prokopovych2011}
Prokopovych, P. (2011).
\newblock On equilibrium existence in payoff secure games.
\newblock {\em Economic Theory}, 48(1):5--16.

\bibitem[Reny, 1999]{reny1999}
Reny, P. (1999).
\newblock On the existence of pure and mixed strategy {Nash} equilibria in
  discontinuous games.
\newblock {\em Econometrica}, 67(5):1029--1056.

\bibitem[Scalzo, 2020]{scalzo2020existence}
Scalzo, V. (2020).
\newblock Existence of strong equilibrium in discontinuous games.
\newblock {\em Journal of Mathematical Analysis and Applications},
  491(1):124247.

\bibitem[Shapley, 1973]{shapley1973balanced}
Shapley, L.~S. (1973).
\newblock On balanced games without side payments.
\newblock In {\em Mathematical programming}, pages 261--290. New York: Academic
  Press.

\bibitem[Song, 2024]{song2024structure}
Song, Q.-Q. (2024).
\newblock On the structure of core solutions of discontinuous general
  cooperative games.
\newblock {\em International Journal of Game Theory}, 53(1):105--125.

\bibitem[Song et~al., 2024]{song2024alpha}
Song, Q.-Q., Guo, M., and Chi, X.-Y. (2024).
\newblock The {$\alpha$}-core in a multi-objective game with set payoffs.
\newblock {\em Mathematical Social Sciences}, 131:32--39.

\bibitem[Uyan{\i}k, 2015]{uyanik2015}
Uyan{\i}k, M. (2015).
\newblock On the nonemptiness of the $\alpha$-core of discontinuous games:
  Transferable and nontransferable utilities.
\newblock {\em Journal of Economic Theory}, 158:213--231.

\bibitem[Yang and Wang, 2018]{yang2018fuzzygames}
Yang, Z. and Wang, A. (2018).
\newblock Existence and stability of the $\alpha$-core for fuzzy games.
\newblock {\em Fuzzy Sets and Systems}, 341:59--68.

\bibitem[Zhao, 1999]{zhao1999b}
Zhao, J. (1999).
\newblock The existence of {TU} $\alpha$-core in normal form games.
\newblock {\em International Journal of Game Theory}, 28(1):25--34.

\end{thebibliography}




%
\end{document}